\documentclass{amsart}
\usepackage{amsfonts}
\usepackage{amsmath}
\usepackage{amssymb}
\usepackage{amsthm}

\usepackage{graphicx}
\usepackage{capt-of}
\usepackage{tikz}
\usetikzlibrary{patterns}
\usetikzlibrary{decorations.markings}
\usepackage{multirow}
\usepackage{array} 
\newcolumntype{C}{>{$}c<{$}}

\newtheorem{theorem}{Theorem}
\newtheorem{prop}{Proposition}
\newtheorem{lemma}{Lemma}
\newtheorem{rem}{Remark}
\newtheorem{exmp}{Example}

\begin{document}

\title[A Wigner-type theorem in terms of equivalent pairs of subspaces]{A non-surjective Wigner-type theorem in terms of equivalent pairs of subspaces}
\author{Mark Pankov}
\address{Faculty of Mathematics and Computer Science, University of Warmia and Mazury, S{\l}oneczna 54, Olsztyn, Poland}
\email{pankov@matman.uwm.edu.pl}

\maketitle

\begin{abstract}
Let $H$ be an infinite-dimensional complex Hilbert space and 
let ${\mathcal G}_{\infty}(H)$ be the set of all closed subspaces of $H$ whose dimension and codimension both are infinite.
We investigate (not necessarily surjective) transformations of ${\mathcal G}_{\infty}(H)$ sending every pair of subspaces to 
an equivalent pair of subspaces; two pairs of subspaces are equivalent if there is a linear isometry sending one of these pairs to the other.
Let $f$ be such a transformation. We show that there is a unique up to a scalar multiple linear or conjugate linear isometry $L:H\to H$ such that
for every $X\in {\mathcal G}_{\infty}(H)$ the image $f(X)$ is the sum of $L(X)$ and a certain closed subspace $O(X)$ orthogonal to the range of $L$.
In the case when $H$ is separable, we give the following sufficient condition to assert  that $f$ is induced by a linear or conjugate linear isometry:
if $O(X)=0$ for a certain  $X\in {\mathcal G}_{\infty}(H)$, then the same holds for all $X\in {\mathcal G}_{\infty}(H)$.
\end{abstract}

\section{Introduction}
The non-bijective version of Wigner's theorem states that every (not necessarily surjective) transformation
of the Grassmannian formed by $1$-dimensional subspaces of a complex Hilbert space (the set of pure states)
preserving the angle between subspaces (the transition probability) is induced by a linear or conjugate linear isometry. 
Moln\'ar \cite{Molnar1,Molnar2} (see also \cite{Molnar-book}) extended this statement onto Grassmannians of finite-dimensional subspaces:
every transformation of such a Grassmannian preserving the principal angles between subspaces is induced by a linear or conjugate linear isometry
(except in the case when the Grassmannian is formed by subspaces whose dimension is the half of the dimension of the corresponding Hilbert space, 
see  \cite{Molnar2} for the details).
Some generalizations of this result were obtained by Geh\'er \cite{Geher} and \v{S}emrl \cite{Semrl}.

It was noted in \cite[p. 66]{Molnar-book} that the principal angles can be defined for pairs of infinite-dimensional closed subspaces.
For our purpose it is sufficient to know that  the principal angles are equal for two pairs of subspaces if and only if 
there is a linear isometry sending one of these pairs to the other  \cite[p.78]{Molnar-book};
in this case, we say that the pairs are equivalent.

In the present paper, we consider  the Grassmannian  formed by all closed subspaces whose dimension and codimension both are infinite
and (not necessarily surjective) transformations of this Grassmannian which send every pair of subspaces to an equivalent pair.
For every such transformation there is a unique up to a scalar multiple linear or  conjugate linear isometry $L$ such that for every $X$ 
(belonging to the Grassmannian)
the image of $X$ is the sum of $L(X)$ and a certain closed subspace $O(X)$ orthogonal to the range of $L$ (Theorem \ref{th1});
note that $O(X)=O(Y)$ if $X,Y$ belong to the same component of the Grassmannian, i.e. the codimensions of $X\cap Y$ 
in $X,Y$ both are finite.
In the separable case, we show that $O(X)=0$ for all $X$ if this equality holds for a certain $X$ (Proposition \ref{prop})
and obtain a sufficient condition to assert  that our transformation is induced by a linear or conjugate linear isometry (Theorem \ref{th2}).
In the non-separable case, the Grassmannian contains subspaces of different dimensions, for this reason,
we are not able to prove the latter statement (Remark \ref{rem}).

\section{Results}
Let $H$ be an infinite-dimensional complex Hilbert space.
For every positive integer $k$ we denote by ${\mathcal G}_{k}(H)$ the Grassmannian of $k$-dimen\-sional subspaces of $H$
and write ${\mathcal G}^{k}(H)$ for the Grassmannian formed by closed subspaces of codimension $k$.
Let ${\mathcal G}_{\infty}(H)$ be the Grassmannian of all closed subspaces of $H$ whose dimension and codimension both are infinite.

We say that two pairs of closed subspaces of $H$ are  {\it equivalent} if there is a  linear isometry 
transferring one of these pairs to the other.
This notion is related to the concept of principal angles between subspaces  \cite[p.78]{Molnar-book}.

Let ${\mathcal G}$ be one of the above Grassmannians. 
If $f$ is a transformation of ${\mathcal G}$ sending every pair of subspaces to an equivalent pair of subspaces, then for any $X,Y\in {\mathcal G}$
there is a linear isometry $L_{\{X,Y\}}:H\to H$ such that 
$$\{f(X),f(Y)\}=\{L_{\{X,Y\}}(X),L_{\{X,Y\}}(Y)\};$$
the latter does not imply that $f(X)=L_{\{X,Y\}}(X)$ and $f(Y)=L_{\{X,Y\}}(Y)$.

By Moln\'ar \cite{Molnar1}, every such transformation of ${\mathcal G}_{k}(H)$
is induced by a linear or conjugate linear isometry of $H$
(for $k=1$ this follows from the non-bijective version of Wigner's theorem). 
Suppose that $f$ is a transformation of ${\mathcal G}^{k}(H)$ which sends every pair of subspaces to an equivalent pair of subspaces.
The same holds for the transformation of ${\mathcal G}_k(H)$ sending every $X\in {\mathcal G}_k(H)$ to $(f(X^{\perp}))^{\perp}$
and this transformation is induced by a linear or conjugate linear isometry $L:H\to H$.
Then for every $Y\in {\mathcal G}^{k}(H)$ we have $$f(Y)=L(Y)\oplus O,$$
where $O$ is the orthogonal complement of the range of $L$;
since the codimension of $L(Y)$ in the range of $L$ is equal to $k$, we obtain that $O=0$ and, consequently, $L$ is a unitary or anti-unitary operator on $H$.

A subset ${\mathcal C}\subset {\mathcal G}_{\infty}(H)$ is called a {\it component} of ${\mathcal G}_{\infty}(H)$ if
\begin{equation}\label{eq-comp}
{\rm codim}_{X}(X\cap Y)<\infty,\;\;\;{\rm codim}_{Y}(X\cap Y)<\infty
\end{equation}
for all $X,Y\in {\mathcal C}$ and ${\mathcal C}$  is maximal with respect to this property.
For each $X\in {\mathcal G}_{\infty}(H)$ there is a unique component containing $X$; it is formed by all $Y\in {\mathcal G}_{\infty}(H)$ satisfying \eqref{eq-comp}.

\begin{exmp}{\rm
Let $L:H\to H$ be a linear or conjugate linear isometry and let $O$ be a closed subspace orthogonal to  the range of $L$.
For a component ${\mathcal C}\subset {\mathcal G}_{\infty}(H)$
consider the map sending every $X\in {\mathcal C}$ to $L(X)\oplus O$.
This is a map  into a  component of ${\mathcal G}_{\infty}(H)$
transferring  every pair of subspaces to an equivalent pair of subspaces.
}\end{exmp}

\begin{theorem}\label{th1}
Let $f$ be a transformation of ${\mathcal G}_{\infty}(H)$ sending  every pair of subspaces to an equivalent pair of subspaces.
Then there is a  unique up to a scalar multiple linear or conjugate linear isometry $L:H\to H$ and for every $X\in {\mathcal G}_{\infty}(H)$
there is  a closed subspace $O(X)$ orthogonal to the range of $L$
such that 
$$f(X)=L(X)\oplus O(X).$$
If $X,Y$ belong to the same component of ${\mathcal G}_{\infty}(H)$, then $O(X)=O(Y)$.
Also, if $X,Y\in {\mathcal G}_{\infty}(H)$ and $Y\subset X$, then $O(Y)\subset O(X)$.
\end{theorem}

Let $f$ be as in Theorem \ref{th1}. If $X\in {\mathcal G}_{\infty}(H)$, then
$f({\mathcal G}^1(X))$ consists of all closed hyperplanes of $f(X)$ containing $O(X)$.
Therefore, 
\begin{equation}\label{eq-th2}
f({\mathcal G}^1(X))={\mathcal G}^{1}(f(X))
\end{equation}
if and only if $O(X)=0$.
In the case when $H$ is separable, we show that $O(X)=0$ for all $X\in {\mathcal G}_{\infty}(H)$
if this equality holds for a certain $X\in {\mathcal G}_{\infty}(H)$ (Proposition \ref{prop}).  Theorem \ref{th1} and the above observation imply the following.

\begin{theorem}\label{th2}
If $H$ is separable, $f$ is a transformation of ${\mathcal G}_{\infty}(H)$ sending  every pair of subspaces to an equivalent pair of subspaces
and \eqref{eq-th2} holds for a certain $X\in {\mathcal G}_{\infty}(H)$, then $f$ is induced by a linear or conjugate linear isometry of $H$.  
\end{theorem}

Remark \ref{rem} (at the end of the paper) explains why the same statement  is not proved for the case when $H$ is non-separable.

\begin{rem}{\rm
Under the assumption that $H$ is separable, \v{S}emrl \cite{Semrl} constructs a non-surjective transformation of 
${\mathcal G}_{\infty}(H)$ which preserves the infimum of principal angles and does not preserve the inclusion relation.
This transformation does not satisfy the condition that all pairs of subspaces go to equivalent pairs of subspaces.
}\end{rem}

\section{Proofs}
\subsection{Proof of Theorem \ref{th1}}
Let $f$ be a transformation of ${\mathcal G}_{\infty}(H)$ sending  every pair of subspaces to an equivalent pair of subspaces.
The following assertions are obvious:
\begin{enumerate}
\item[$\bullet$] $f$ is injective;
\item[$\bullet$] $f$ is orthogonality, compatibility  and ortho-adjacency preserving in both directions.
\end{enumerate}
Recall that two closed subspaces of $H$ are {\it compatible} if there is an orthonormal basis of $H$ such that each of these subspaces is spanned by a subset of this basis.
Compatible subspaces $X$ and $Y$ belonging to the same Grassmannian are called {\it ortho-adjacent}
if $X\cap Y$ is of codimension $1$ in both $X,Y$ \cite{PPZ, PT}.

\begin{lemma}\label{lemma0}
The transformation $f$ is inclusion preserving in both directions.
\end{lemma}

\begin{proof}
Let $X,Y\in {\mathcal G}_{\infty}(H)$.
Suppose that $X\subset Y$. Then one of the subspaces $f(X),f(Y)$ is contained in the other. 
We take any $Z\in {\mathcal G}_{\infty}(H)$ orthogonal to $X$ and non-orthogonal to $Y$.
Then $f(Z)$ is orthogonal to $f(X)$ and non-orthogonal to $f(Y)$ which means that the inclusion $f(Y)\subset f(X)$ is impossible. 
Therefore, $f(X)\subset f(Y)$.

Conversely, if $f(X)\subset f(Y)$, then one of the subspaces $X,Y$ is contained in the other and the above arguments show that $X\subset Y$.
\end{proof}

It is clear that for any $X,Y\in {\mathcal G}_{\infty}(H)$ such that $Y\subset X$ we have
$${\rm codim}_X(Y)={\rm codim}_{f(X)}(f(Y)).$$
For every $X\in {\mathcal G}_{\infty}(H)$
we denote by ${\mathcal G}^{\rm fin}(X)$ 
the set  formed by all closed subspaces of $X$ whose codimension in $X$  is finite.

\begin{lemma}\label{lemma1}
For every $X\in {\mathcal G}_{\infty}(H)$ 
there is a unique up to a scalar multiple linear or conjugate-linear isometry $L_{X}:X\to f(X)$  such that
for every $Y\in {\mathcal G}^{\rm fin}(X)$ we have
$$f(Y)=L_{X}(Y)\oplus O(X),$$ 
where $O(X)=(L_X(X))^{\perp}\cap f(X)$ is the orthogonal complement of the range of $L_X$ in $f(X)$.
Furthermore, the following assertions are fulfilled:
\begin{enumerate}
\item[$\bullet$] if $Y\in {\mathcal G}_{\infty}(H)$ is contained in $X$, then $L_Y$ is a scalar multiple of the restriction of $L_X$ to $Y$ and $O(Y)\subset O(X)$;
\item[$\bullet$] $O(X)=O(Y)$ for every $Y$ belonging to the component of ${\mathcal G}_{\infty}(H)$ containing $X$.
\end{enumerate}
\end{lemma}

\begin{proof}
(1). Suppose that $f(X)=X'$.
Denote by ${\mathcal X}$ and ${\mathcal X}'$ the sets consisting of all elements of ${\mathcal G}_{\infty}(H)$ contained in $X$ and $X'$,
respectively. 
Let ${\mathcal X}^{\perp}$ and ${\mathcal X}'^{\perp}$ be the sets formed by the orthogonal complements of 
elements from  ${\mathcal X}$ and ${\mathcal X}'$ in $X$ and $X'$, respectively.
Consider  the map 
$g:{\mathcal X}^{\perp}\to {\mathcal X}'^{\perp}$ which sends every $Z\in {\mathcal X}^{\perp}$ to 
$$(f(Z^{\perp}\cap X))^{\perp}\cap X'.$$
This map is inclusion preserving in both directions (since $f$ is inclusion preserving in both directions);
furthermore, for every positive integer $k$ we have
$$f({\mathcal G}^k(X))\subset {\mathcal G}^k(X')$$
which implies that 
\begin{equation}\label{eq-g}
g({\mathcal G}_k(X))\subset {\mathcal G}_k(X').
\end{equation}
Observe that $A,B$ belonging to ${\mathcal G}^{1}(X)$ or ${\mathcal G}^{1}(X')$ are ortho-adjacent if and only if
$A^{\perp}\cap X, B^{\perp}\cap X$ or $A^{\perp}\cap X', B^{\perp}\cap X'$ 
(belonging to ${\mathcal G}_1(X)$ or ${\mathcal G}_1(X')$, respectively) are orthogonal.
Since $f$ is ortho-adjacency preserving, the restriction of $g$ to ${\mathcal G}_1(X)$ is orthogonality preserving.
Taking $k=2$ in \eqref{eq-g} we establish that $g|_{{\mathcal G}_1(X)}$ sends lines of the projective space associated to $X$
into lines of the projective space associated to $X'$.
Since the image of $g|_{{\mathcal G}_1(X)}$ is not contained in a line, 
$g|_{{\mathcal G}_1(X)}$ is induced by a semilinear injection of $X$ to $X'$ 
(Faure-Fr\"{o}licher-Havlicek's version of the Fundamental Theorem of Projective Geometry \cite{FF,Havlicek1}, see also \cite[Theorem 2.6]{Pankov-book}).
The transformation $g$ is orthogonality preserving and, consequently, this semilinear injection sends orthogonal vectors to orthogonal vectors which means 
that it is a scalar multiple of a linear or conjugate linear isometry \cite[Proposition 4.2]{Pankov-book}.
So, there is  a linear or conjugate-linear isometry $L_{X}:X\to X'$
such that 
$$g(Z)=L_X(Z)$$
for every $1$-dimensional subspace $Z\subset X$.

The same holds for every finite-dimensional subspace $Z\subset X$.
Indeed, 
$$L_X(P)=g(P)\subset g(Z)$$ for every $1$-dimensional subspace $P\subset Z$
which implies that $L_X(Z)\subset g(Z)$; therefore, $L_X(Z)$ coincides with $g(Z)$, since these subspaces are of the same finite dimension.

Then for every $Y\in {\mathcal G}^{\rm fin}(X)$ we have
$$f(Y)=(g(Y^{\perp}\cap X))^{\perp}\cap X'=(L_{X}(Y^{\perp}\cap X))^{\perp}\cap X'=L_{X}(Y)\oplus O(X),$$
where $O(X)=(L_X(X))^{\perp}\cap X'$ is the orthogonal complement of the range of $L_X$ in $X'$.
Observe that 
$$O(X)=\bigcap_{Y\in {\mathcal G}^{\rm fin}(X)}f(Y).$$
If there is a linear or conjugate linear isometry $L':X\to X'$ and a closed subspace $O'\subset X'$ orthogonal to the range of $L'$
such that 
$$f(Y)=L'(Y)\oplus O'$$
for all $Y\in {\mathcal G}^{\rm fin}(X)$, 
then 
$$O'=\bigcap_{Y\in {\mathcal G}^{\rm fin}(X)}f(Y)$$
and, consequently, $O(X)=O'$.
Since $O(X)=O'$ is orthogonal to the ranges of $L_X$ and $L'$,
we obtain that $L_X(Y)=L'(Y)$
for all $Y\in {\mathcal G}^{\rm fin}(X)$.

For every $1$-dimensional subspace $P\subset X$ 
the $1$-dimensional subspaces $L_{X}(P)$ and $L'(P)$ are the orthogonal complements  of $L_X(P^{\perp}\cap X)$ and $L'(P^{\perp}\cap X)$ (respectively) in $X'$.
The equality $L_X(P^{\perp}\cap X)=L'(P^{\perp}\cap X)$ implies that $L_{X}(P)=L'(P)$.
This guarantees that $L'$ is a scalar multiple of $L_X$.

(2). Our next step is to show that for every $Y\in {\mathcal  X}$ whose codimension in $X$ is infinite
there is a closed subspace $C(Y)$ orthogonal to the range of $L_X$ and such that
$$f(X)=L_X(Y)\oplus C(Y).$$
It will be shown latter that $C(Y)=O(Y)$.

For every $1$-dimensional subspace $P\subset Y^{\perp}\cap X$ we have 
$$L_X(P)=g(P)\subset g(Y^{\perp}\cap X)$$
which implies that 
$$L_X(Y^{\perp}\cap X)\subset g(Y^{\perp}\cap X)$$
and 
$$g(Y^{\perp}\cap X)=L_X(Y^{\perp}\cap X)\oplus A,$$
where $A$ is the orthogonal complement of $L_X(Y^{\perp}\cap X)$ in $g(Y^{\perp}\cap X)$.
We assert that $A$ is orthogonal to the range of $L_X$.

Let $P$ be a $1$-dimensional subspace of $Y$.
The $P^{\perp}\cap X$ and $Y$ are compatible and $Y$ is not contained in $P^{\perp}\cap X$.
Then $f(P^{\perp}\cap X)$ and $f(Y)$ are compatible and $f(Y)$ is not contained in $f(P^{\perp}\cap X)$.
The latter implies that the $1$-dimensional subspace
$$L_X(P)=g(P)=(f(P^{\perp}\cap X))^{\perp}\cap X'$$
is orthogonal to 
$$g(Y^{\perp}\cap X)=(f(Y))^{\perp}\cap X'.$$
Therefore, every $1$-dimensional subspace of $L_X(Y)$ is orthogonal to $g(Y^{\perp}\cap X)$
which means that $L_X(Y)$ is orthogonal to $A\subset g(Y^{\perp}\cap X)$.
So, $A$ is orthogonal to both $L_X(Y)$ and $L_X(Y^{\perp}\cap X)$
whose sum is  the range of $L_X$.

Since
$$g(Y^{\perp}\cap X)=L_X(Y^{\perp}\cap X)\oplus A$$
is contained in 
$$X'=f(X)=L_{X}(X)\oplus O(X)$$
and $A,O(X)$ are orthogonal to the range of $L_X$,
we obtain that $A\subset O(X)$; furthermore, 
$$f(Y)=(g(Y^{\perp}\cap X))^{\perp}\cap X' =L_{X}(Y)\oplus (A^{\perp}\cap O(X))$$
and $C(Y)=A^{\perp}\cap O(X)$ is as required.

(3).
Now, let $Y$ be an arbitrary element of ${\mathcal X}$.

We take a $1$-dimensional subspace $P\subset Y$ and denote by $Z$ the orthogonal complement of $P$ in $X$. 
Then $Z\cap Y$ is the orthogonal complement of $P$ in $Y$.
Observe  that $L_X(P)$ is the orthogonal  complement of 
$$f(Z)=L_X(Z)\oplus O(X)\;\mbox{ in }\;f(X)=L_X(X)\oplus O(X)$$
and $L_Y(P)$ are the orthogonal  complement of 
$$f(Z\cap Y)=L_Y(Z\cap Y)\oplus O(Y)\;\mbox{ in }\; f(Y)=L_Y(Y)\oplus O(Y).$$
It was established above that
$$f(Y)=L_X(Y)\oplus A,$$
$$f(Z\cap Y)=L_X(Z\cap Y)\oplus B,$$
where $A,B$ are closed subspaces orthogonal to the range of $L_X$
(we have  $A=B=O(X)$ if $Y\in {\mathcal G}^{\rm fin}(X)$ and $A=C(Y),B=C(Z\cap Y)$ if the codimension of $Y$ in $X$ is infinite).
Since $f$ is inclusion preserving and the subspaces $A,B,O(X)$ are orthogonal to the range of $L_X$, we obtain that $B\subset A\subset O(X)$.
The codimension of $f(Z\cap Y)$ in  $f(Y)$ and the codimension of $L_X(Z\cap Y)$ in $L_X(Y)$ 
both are equal to $1$ which guarantees that $A=B$.
This means that 
$$f(Z)=L_X(Z)\oplus O(X)$$ 
intersects 
$$f(Y)=L_X(Y)\oplus A$$
precisely in 
$$f(Z\cap Y)=L_X(Z\cap Y)\oplus A.$$
Recall that $L_X(P)$ is the orthogonal complement of $f(Z)$ in $f(X)$ and $L_Y(P)$ is the orthogonal complements of $f(Z\cap Y)$ in $f(Y)$.
Each of the $1$-dimensional subspaces $L_X(P),L_Y(P)$ is contained in $f(Y)$ and, consequently,
these $1$-dimensional subspaces are coincident.

So, $L_X(P)=L_Y(P)$ for every $1$-dimensional subspace $P\subset Y$ which means that the restriction of $L_X$ to $Y$ is a scalar multiple of $L_Y$.

If $Y\in {\mathcal G}^{\rm fin}(X)$, then
$$L_{X}(Y)\oplus O(X)=f(Y)=L_Y(Y)\oplus O(Y)=L_X(Y)\oplus O(Y)$$
which means that $O(X)=O(Y)$ (since  $O(X),O(Y)$ both  are orthogonal to $L_X(Y)$).
As a consequence, we obtain that 
$$O(X)=O(X\cap Y')=O(Y')$$
for every $Y'$ belonging to the component of ${\mathcal G}_{\infty}(H)$ containing $X$.

If the codimension of $Y$ in $X$ is infinite, then
$$L_Y(Y)\oplus C(Y)=L_X(Y)\oplus C(Y)=f(Y)=L_Y(Y)\oplus O(Y)$$
which implies that $O(Y)=C(Y)$
(since $C(Y),O(Y)$ both are orthogonal to the range of $L_Y$) and, consequently, $O(Y)=C(Y)\subset O(X)$.
\end{proof}

\begin{lemma}\label{lemma2}
There is a linear or a conjugate-linear isometry $L:H\to H$ such that 
for every $X\in {\mathcal G}_{\infty}(H)$ the restriction of $L$ to $X$ is a scalar multiple of $L_X$
and $O(X)$ is orthogonal to the range of $L$.
\end{lemma}

\begin{proof}
For a $1$-dimensional subspace $P\subset H$ consider distinct  $X,Y\in {\mathcal G}_{\infty}(H)$ containing $P$.
If $X\cap Y$ belongs to ${\mathcal G}_{\infty}(H)$, then 
$$L_{X}(P)=L_{X\cap Y}(P)=L_{Y}(P)$$
(by Lemma \ref{lemma1}).
In the case when $X\cap Y$ is finite-dimensional, 
there is $Z\in{\mathcal G}_{\infty}(H)$ such that $X\cap Z$ and $Y\cap Z$ both belong to ${\mathcal G}_{\infty}(H)$
\cite[Lemma 3.20]{Pankov-book} and Lemma \ref{lemma1} implies that
$$L_X(P)=L_{X\cap Z}(P)=L_Z(P)=L_{Y\cap Z}(P)=L_Y(P).$$
Consider the transformation $h$ of ${\mathcal G}_1(H)$ defined as follows: 
for every $1$-dimensional subspace $P\subset H$ we take any $X\in {\mathcal G}_{\infty}(H)$ containing $P$ and set $h(P)=L_X(P)$
(it was established above that $h(P)$ does not depend on the choose of $X$).
This is a transformation of the projective space associated to $H$ sending lines into lines and
the image of $h$ is not contained in a line;
furthermore, $h$ is orthogonality preserving. 
This means that $h$ is induced by a linear or conjugate-linear isometry $L:H\to H$
\cite[Theorem 2.6 and Proposition 4.2]{Pankov-book}.
By our construction, the restriction of $L$ to any $X\in {\mathcal G}_{\infty}(H)$ is a scalar multiple of $L_X$. 

For every $X\in {\mathcal G}_{\infty}(H)$ the subspaces
$$f(X)=L(X)\oplus O(X)\;\mbox{ and }\; f(X^{\perp})=L(X^{\perp})\oplus O(X^{\perp})$$
are orthogonal. 
In particular, $O(X)$ is orthogonal to $L(X^{\perp})$.
Since the range of $L$ is the sum of $L(X)$ and $L(X^{\perp})$,
the subspace $O(X)$ is orthogonal to the range of $L$.
\end{proof}

Theorem \ref{th1} follows immediately from Lemmas \ref{lemma1} and \ref{lemma2}.

\subsection{Proof of Theorem \ref{th2}}
As above, we assume that $f$ is a transformation of ${\mathcal G}_{\infty}(H)$ sending every pair of subspaces to an equivalent pair of subspaces.
By Theorem \ref{th1}, for $X\in {\mathcal G}_{\infty}(H)$ we have
$$
f({\mathcal G}^1(X))={\mathcal G}^{1}(f(X))
$$
if and only if $O(X)=0$.
Theorem \ref{th2} is a consequence of  the following.

\begin{prop}\label{prop}
If $H$ is separable and $O(X)=0$ for a certain $X\in {\mathcal G}_{\infty}(H)$,  then the same holds for all $X\in {\mathcal G}_{\infty}(H)$.
\end{prop}

To prove Proposition \ref{prop} we use the following lemma.

\begin{lemma}\label{lemmaZ}
If $H$ is separable, then
for every $X\in {\mathcal G}_{\infty}(H)$ there is $Q\in {\mathcal G}_{\infty}(H)$ such that each of $Q,Q^{\perp}$ intersects $X^{\perp}$ precisely in $0$.
\end{lemma}

\begin{proof}
Let $\{x_i\}_{i\in {\mathbb N}}$ and $\{x'_i\}_{i\in {\mathbb N}}$ be orthonormal bases of $X$ and $X^{\perp}$, respectively.
Denote by $Q$ the closed subspace whose orthonormal basis is $\{\frac{1}{\sqrt{2}}(x_i+x'_i)\}_{i\in {\mathbb N}}$. 
Then  $\{\frac{1}{\sqrt{2}}(x_i-x'_i)\}_{i\in {\mathbb N}}$ is an orthonormal basis of $Q^{\perp}$. 
A direct verification shows that $Q$ is as required. 
\end{proof}

\begin{proof}[Proof of Proposition \ref{prop}]
Let $X$ be an element of ${\mathcal G}_{\infty}(H)$ satisfying $O(X)=0$.

If $Y\in {\mathcal G}_{\infty}(H)$ intersects $X^{\perp}$ precisely in zero, then $O(Y)=0$.
Indeed, $Y$ does not contain a non-zero subspace orthogonal to $X$
and, consequently, $f(Y)$ does not contain a non-zero subspace orthogonal to $f(X)=L(X)$ which means that $O(Y)=0$
(since $O(Y)$ is orthogonal to the range of $L$).

Let $Q$ be as in Lemma \ref{lemmaZ}.
Then $Q\cap X^{\perp}=0$ which implies that $O(Q)=0$. 
Since $X^{\perp}\cap Q^{\perp}=0$ and $O(Q)=0$, we obtain  that $O(X^{\perp})=0$.

By the same arguments, $O(Y^{\perp})=0$ for every $Y\in {\mathcal G}_{\infty}(H)$ satisfying $O(Y)=0$.

If $Y\in {\mathcal G}_{\infty}(H)$ is contained in $X$, then $O(Y)=0$ (we have $O(Y)\subset O(X)$ by Theorem \ref{th1}).
If $Y\in {\mathcal G}_{\infty}(H)$ contains $X$, 
then $Y^{\perp}$ is contained in $X^{\perp}$;
since $O(X^{\perp})=0$, we obtain that $O(Y^{\perp})=0$. 
Then $O(Y)=O(Y^{\perp\perp})=0$.

The same arguments show that if $Y\in {\mathcal G}_{\infty}(H)$ satisfies $O(Y)=0$ and $Y'\in {\mathcal G}_{\infty}(H)$ is contained in $Y$ or contains $Y$,
then $O(Y')=0$.

Let $Y$ be an arbitrary element of ${\mathcal G}_{\infty}(H)$. 
If $X\cap Y$ belongs to ${\mathcal G}_{\infty}(H)$, 
then $O(X\cap Y)=0$ which implies that $O(Y)=0$.
In the case when $X\cap Y$ is finite-dimensional,
there is  $Z\in {\mathcal G}_{\infty}(H)$ such that $X\cap Z$ and $Y\cap Z$ both belong to ${\mathcal G}_{\infty}(H)$ \cite[Lemma 3.20]{Pankov-book}.
Then $O(X\cap Z)=0$ and, consequently, $O(Z)=0$.
Therefore, $O(Y\cap Z)=0$ which means that $O(Y)=0$.
\end{proof}

\begin{rem}\label{rem}{\rm
Suppose that $H$ is non-separable and 
for $X\in {\mathcal G}_{\infty}(H)$ there is $Q\in {\mathcal G}_{\infty}(H)$ such that both $Q,Q^{\perp}$ intersect $X^{\perp}$ precisely in $0$.
Then the orthogonal projections of $Q$ and $Q^{\perp}$ on $X$ are injective and, by \cite[Problem 56]{Halmos}, 
the dimensions of $Q$ and $Q^{\perp}$ are not greater than the dimension of $X$.
The latter is impossible if the dimension of $X$ is less than the dimension of $H$
(which happens in the non-separable case).
Therefore, if $H$ is non-separable, then Lemma \ref{lemmaZ} fails and  we are not able to prove Proposition \ref{prop}.
}\end{rem}

\end{document}